\documentclass[a4paper,UKenglish]{article}

\usepackage{microtype}
\usepackage{fullpage}
\usepackage{graphics}
\graphicspath{{./figures/},{./figures/}}

\usepackage{authblk}

\usepackage{graphicx}
\usepackage{amsthm}
\usepackage{amsmath}
\usepackage{enumerate}
\usepackage{color}
\usepackage{xspace}
\usepackage{url}

\usepackage{amsfonts}\usepackage{amssymb}
\usepackage{thmtools}\usepackage{mathscinet}
\usepackage{thm-restate}
    \newtheorem{theorem}{Theorem}
    
    \newtheorem{lemma}[theorem]{Lemma}
    
    \newtheorem{remark}[theorem]{Remark}

    	\definecolor{darkgreen}{rgb}{0.01, 0.93, 0.29}
\definecolor{lightbrown}{rgb}{0.91, 0.4, 0.11}
\usepackage{framed}

\title{Simultaneous Embedding of Colored Graphs\thanks{This work is supported in part by Natural Sciences and Engineering Research Council of Canada (NSERC).}}

\author[1]{Debajyoti Mondal} 
\affil[1]{Department of Computer Science, University of Saskatchewan, Saskatoon, Canada\\
  \texttt{d.mondal@usask.ca}}

\usepackage[textsize=tiny]{todonotes}
\usepackage{verbatim}

\usepackage{setspace}
\newcommand{\J}[1]{{\color{black} #1}}

\begin{document}   

\maketitle  

\begin{abstract}
A set of colored graphs are compatible, if for every  color $i$, the number of vertices of color $i$ is the same in every graph. A simultaneous embedding of $k$ compatibly colored graphs, each with $n$ vertices,  consists of $k$ planar polyline drawings of these graphs such that the vertices of the same color are mapped to a common set of vertex locations. 

We prove that simultaneous embedding of $k\in o(\log \log n)$ colored planar graphs, each with $n$ vertices, can always be computed with a sublinear number of bends per edge. Specifically,   we show an $O(\min\{c, n^{1-1/\gamma}\})$  upper bound on the number of bends per edge, where $\gamma = 2^{\lceil k/2 \rceil}$ and  $c$ is the total number of colors. Our bound, which results from a better analysis of a previously known algorithm [Durocher and Mondal, SIAM J. Discrete Math., 32(4), 2018], improves the bound for $k$, as well as the bend complexity by a factor of $\sqrt{2}^{k}$. The algorithm can be generalized to obtain small universal point sets for colored graphs. We prove that $n\lceil c/b \rceil$ vertex locations, where $b\ge 1$, suffice to embed any set of compatibly colored $n$-vertex planar graphs with bend complexity $O(b)$, where $c$ is the number of colors. 

 
 
\end{abstract}


\section{Introduction}
\label{sec:introduction}
Let $\mathcal{G}$ be a set of $k$ planar graphs such that each graph $G\in \mathcal{G}$ has $n$ vertices and these vertices are labelled without repetition with the numbers $1,\ldots,n$. A \emph{simultaneous embedding} of $\mathcal{G}$ is a set of $k$ planar polyline drawings of the graphs in $\mathcal{G}$ such that the vertices with the same label have the same location in $\mathbb{R}^2$ (Fig.~\ref{fig:intro}(a)--(c)). Simultaneous embedding  can be used to model multilevel circuit  layout~\cite{DBLP:conf/cccg/TahmasbiH10} \J{and to visualize} different types of relations among a common set of nodes, e.g., different types of code clones over the files of a software~\cite{DBLP:journals/vi/MondalMRSLW19}. 
 
\begin{figure}[h]
\includegraphics[width=\textwidth]{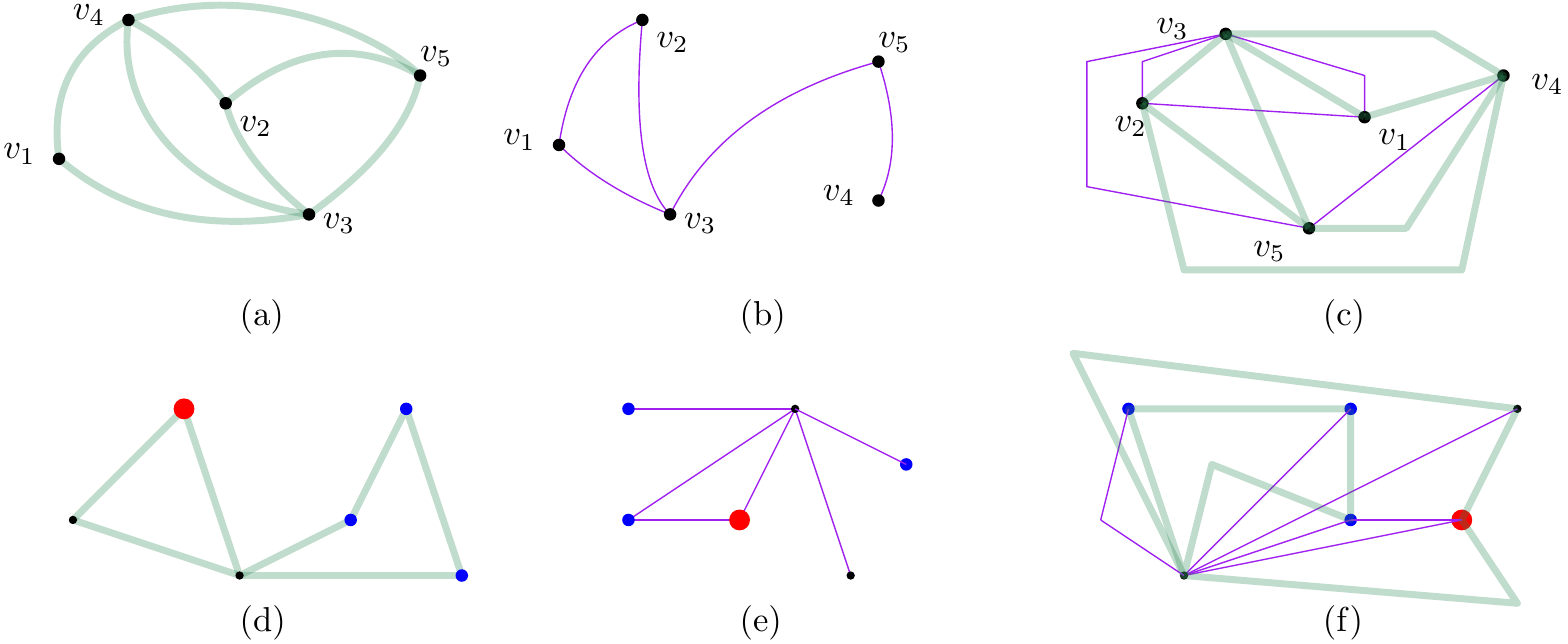}
\caption{(a)--(b) A pair of graphs, and (c) their simultaneous embedding with bend complexity 2. (d)--(f) Illustration for a colored simultaneous embedding of two graphs with bend complexity 1. Here the vertices are colored with red, blue,  and black (i.e., large, small, and tiny discs).} 
\label{fig:intro}
\end{figure}

In a colored simultaneous embedding problem  (Fig.~\ref{fig:intro}(d)--(f)),  \J{the input consists of graphs whose vertices are labelled with the colors $1,\ldots,c$, where $1\le c\le n$,}  such that for every color $q$ and every pair of graphs $G,G'$ in $\mathcal{G}$, the number of vertices of color $q$ in $G$ is the same as that of $G'$, and the locations of the  vertices of color $q$ in the drawing of $G$ are the same as that of $G'$. From both the perspective of readability and VLSI applications, a desirable goal is to minimize the \emph{bend complexity}, i.e., the number of bends per edge, in the drawing. A \emph{$b$-bend simultaneous embedding} consists of drawings with bend complexity at most $b$.

\subsection{Related Research}
Colored simultaneous embedding was first considered by Brandes et al.~\cite{DBLP:journals/algorithmica/BrandesEEFFGGHKKa11}, where they examined classes of graphs that admit simultaneous embedding with bend complexity 0, i.e.,  each edge is a straight line segment, which is also   known as \emph{simultaneous geometric  embedding}. Although  every pair of paths  admits  simultaneous geometric embedding, there exist 3 paths that do not admit simultaneous embedding~\cite{BrassCDEEIKLM07}, even when  colored with $6$ colors~\cite{DBLP:journals/algorithmica/BrandesEEFFGGHKKa11}. However, 
 any number of 3-colored paths can be simultaneously embedded with bend complexity 0~\cite{DBLP:journals/algorithmica/BrandesEEFFGGHKKa11}. A rich body of literature examines   simultaneous embedding   and its  variants (e.g., see the survey by  Bl{\"{a}}sius et al.~\cite{DBLP:reference/crc/BlasiusKR13}). 

The \emph{thickness} of a graph $G$ is the minimum number $t$ such that $G$ can be decomposed into $t$ planar subgraphs. By F\'{a}ry's theorem~\cite{fary}, every planar graph (equivalently, thickness-1 graph) has a planar straight-line drawing. Simultaneous embedding of $k$ planar graphs with small bend complexity can be seen as an extension of F\'{a}ry's theorem  for drawing thickness-$k$ graphs on $k$ planar layers. Every pair of planar graphs can be  simultaneously embedded using bend complexity 2~\cite{EK05,GiacomoL07}, and thus thickness-2 graphs admit 2-bend polyline drawings on 2 planar layers. 

Durocher and Mondal~\cite{DBLP:journals/siamdm/DurocherM18} showed that  every thickness-$k$ graph can be drawn on $k$ planar layers with bend complexity at most  $O(\sqrt{2}^{k} \cdot n^{1-1/\beta})$, where $\beta = 2^{\lceil (k-2)/2 \rceil }$.   Although the bound seems to be sublinear for $k\in 2\log ( o(n^{1/\beta}))$, they claimed a sublinear upper bound only for fixed $k$. This case is straightforward to observe from their algorithm, but the case when $k$ is not fixed, requires a careful analysis of the dependency between $n$ and $k$. In this paper, we modify the algorithm to accommodate the colors, and provide a better analysis of the algorithm leading to a bound of $O(\min\{c, n^{1-1/\gamma}\})$, where $\gamma = 2^{\lceil k/2 \rceil}$. This bound holds for $k\in o(\log \log n)$. In addition to the improved bound on $k$, it also removes the multiplicative factor $\sqrt{2}^k$  from the previously known bound. 

A point set $S$ is \emph{$t$-bend universal} for a class  of graphs if every graph in that class  admits a $t$-bend planar polyline drawing on $S$. There exists an  $1$-bend universal point set of size $n$ for planar  graphs~\cite{DBLP:conf/gd/EverettLLW07}. Note that this can be seen as a universal point set for graphs that are colored with a single color. For outerplanar graphs, which are compatibly colored  with $c$ colors,  a $(4c+1)$-bend universal point set is known~\cite{DBLP:journals/jgaa/GiacomoDLMTW08}. 


Pach and Wenger~\cite{PachW01} showed that $\Omega(n)$ bends are  sometimes necessary to construct a planar polyline drawing of a graph if for every vertex, a vertex location is  specified in the   input. While the graphs here can be seen as colored with $n$ colors, non-trivial lower bound on the bend complexity has recently been achieved also for graphs colored with only three colors~\cite{DBLP:conf/gd/GiacomoGLN17}.  



\subsection{Contribution} We show that every set of $k\in o(\log \log n)$ planar graphs, each with $n$ vertices and compatibly colored with $c$ colors, can be simultaneously embedded with bend complexity $O(\min\{c,n^{1-1/\gamma}\})$, where $\gamma = 2^{\lceil k/2 \rceil}$. 
 Our bound results from a better analysis of a previously known algorithm [Durocher and Mondal, SIAM J. Discrete Math., 32(4), 2018], and  improves the previously known bound for $k$, as well as the bend complexity by a factor of $\sqrt{2}^{k}$.   We also examine the potential trade-off between the number of vertex locations and the bend complexity. We show that for  $n$-vertex planar graphs, which are compatibly colored with $c$ colors, there exists an  $O(b)$-bend universal point set of size $n\lceil c/b \rceil$, where $b\ge 1$. 



The rest of the paper is organized as follows. Section~\ref{sec:tech} \J{introduces}  some notation and preliminary results. Section~\ref{sec:draw} presents our result on  colored simultaneous embedding of $n$-vertex graphs on a set of $n$  vertex locations. Section~\ref{sec:simple} shows a potential trade-off between the bend complexity and the number of vertex locations allowed to compute a simultaneous embedding. Finally, Section~\ref{sec:conclusion} concludes the paper suggesting some open problems. 

\section{Technical Details}
\label{sec:tech}
In this section \J{we introduce} some notation and preliminary results.

A  \emph{monotone topological book embedding}~\cite{DBLP:journals/comgeo/GiacomoDLW05}  of a planar graph $G$ is a planar drawing  of $G$, where the vertices are represented as points on a horizontal line $\ell$, and each edge is drawn as an $x$-monotone polyline between their corresponding end points  such that it does not cross $\ell$ more than once (Fig.~\ref{fig:monotone}(a)--(b)). The line $\ell$ is the \emph{spine} of the embedding, and the crossing points on $\ell$ are the \emph{division vertices}. The path obtained by connecting the vertices (including the division vertices) on $\ell$ in order (e.g., left to right) is the \emph{spinal path} of $G$.

\begin{figure}[h]
\centering
\includegraphics[width=.7\textwidth]{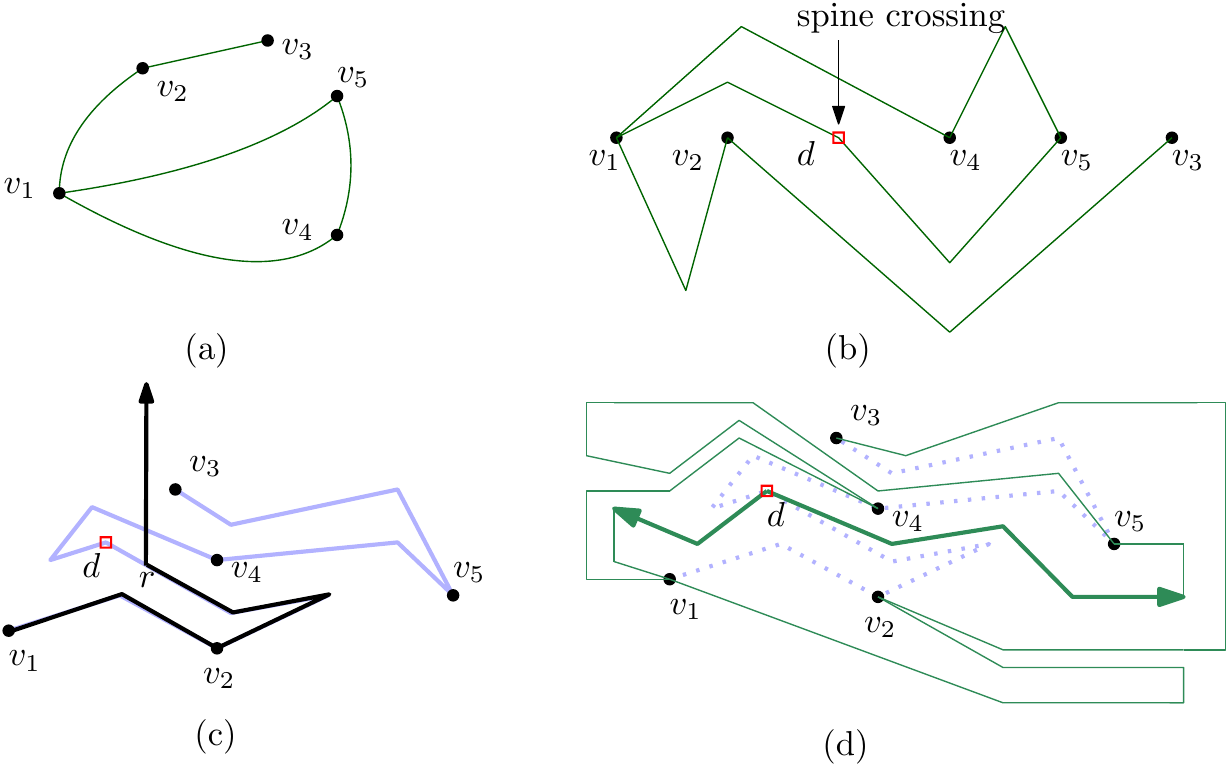}
\caption{(a) A graph $G$, and (b)   its monotone topological book embedding. The  spinal path is $P = (v_1,v_2,d,v_4,v_5,v_3)$, where $d$ is a division vertex. (c)--(d) Illustration for the drawing of $G$ from an uphill drawing of its spinal path. }
\label{fig:monotone}
\end{figure}
Let $\Gamma$ be a planar polyline drawing of a path $P=(v_1,v_2,\ldots,v_n)$. The drawing $\Gamma$ is an \emph{uphill} drawing if for any point $r$ (here $r$ may be a vertex location or an interior point of some edge) in $\Gamma$, the upward ray from $r$ does not intersect the  \J{polyline $v_1,\ldots,v_i,r$}, where $1\le i\le n$ (Fig.~\ref{fig:monotone}(c)).

Two graphs are \emph{compatibly colored} if for each color $q$, both graphs have the same number of vertices of color $q$. Throughout the paper we assume that the graphs are compatibly colored. 
Let $\mathcal{G} = \{G_1,\ldots,G_k\}$ be an instance of the simultaneous embedding problem, where the graphs are colored compatibly with $c$ colors. \J{Consider a monotone topological book embedding of $G_i$ and let $P_i$ be the spinal path.} To color the division vertices in all the graphs compatibly, we introduce dummy vertices as necessary (at the end of the spinal paths). We then color all the division vertices with a color different than the input colors.   Durocher and Mondal~\cite{DBLP:journals/siamdm/DurocherM18} showed that if $P_i$ admits an $b$-bend  uphill drawing on a point set $S$, then $G_i$ admits a  $O(b)$-bend polyline drawing on $S$.
 Hence it suffices to consider only the simultaneous embedding of the spinal paths. Here we give a concise proof for completeness.
 
\begin{lemma}[Durocher and Mondal~\cite{DBLP:journals/siamdm/DurocherM18}]\label{lem:paths}
If a spinal path admits a $b$-bend  uphill drawing on a point set $S$, then the corresponding graph admits an $O(b)$-bend planar polyline drawing on $S$. 
\end{lemma}
\begin{proof}
Let $B$ be the axis-aligned bounding box of $S$.  While drawing the spinal path, at each vertex $v$, draw two $b$-bend polygonal paths starting at $v$. One path hits the left and the other hits the right boundary of $B$ (e.g., see the bold paths starting at $d$ in Fig.~\ref{fig:monotone}(d)). Then draw the rest of the spinal path above these polygonal paths. Finally, the leftward (rightward) polygonal paths are used to draw the edges that are above (below)  the spinal path in the monotone topological book embedding.
\end{proof}
\begin{remark}\label{rem:1}
Let $\mathcal{P}$ be a set of compatibly colored spinal paths and let $S$ be a point set. If each path in  $\mathcal{P}$ admits an uphill embedding on $S$ with bend complexity $O(b)$, then the graphs corresponding to  $\mathcal{P}$ admit a simultaneous embedding with bend complexity $O(b)$.
\end{remark}  
 
Throughout the paper we consider colors as positive integers. Since the drawing algorithm will heavily use results on partitioning an \emph{integer sequence}, i.e., an ordered list of integers, we often consider a colored path as an integer sequence. 

\J{A  sequence (i.e., ordered set) of numbers is \emph{monotonic} if it is either  non-increasing or non-decreasing. A \emph{$k$-tuple} is an ordered set of $k$ numbers. Given an ordered set of $k$-tuples, the \emph{$i$th dimension}, where $1\le i\le k$, is the sequence of integers obtained by taking the $i$th element of every $k$-tuple in the given order. A \emph{sequence of $k$-tuples is monotonic} if each dimension of the sequence is monotonic}. Here we prove two lemmas on partitioning a sequence of \J{integers and $k$-tuples}, which will be used in our drawing algorithm.  
\begin{lemma}\label{lem:recurse}
\J{Assume that for every sequence of $n$ integers (resp., $k$-tuples), one can find a monotonic subsequence of $n^\delta$, where  $0< \delta < 1$.} Then given a sequence $S$ of $n$ integers \J{(resp., $k$-tuples)}, where $n\ge 2^{1/\delta}$,  one can partition $S$ into  $ O(\frac{n^{1-\delta} }{1-\delta})$ disjoint \J{monotonic subsequence}.
\end{lemma}
\begin{proof}
To obtain the required partition, we repeatedly extract a \J{monotonic} sequence of  $n^\delta$ \J{elements} from $S$. 
Thus the number of elements in the partition is determined by the recurrence relation $T(n) = T(n-n^\delta)+1$. We now show that $T(n)$ is upper  bounded by $\frac{dn^{1-\delta}}{1-\delta}$, where  \J{$d =  2(1-\delta) + 1$}.

\J{We choose $m \le 2^{1/\delta}$ as the base case. 
 First observe that in the base case, we have $T(m) = m$, because we can partition $S$ such that each subsequence contains a single \J{element, which is trivially monotonic}. Observe that $m \le 2^{1/\delta} \le  \frac{d2^{1/\delta}}{2(1-\delta)} =   \frac{d2^{(1-\delta)/\delta}}{1-\delta} = \frac{dm^{1-\delta}}{1-\delta}$. This proves the base case.}

\J{We now assume that the inequality $T(k) \le   \frac{dk^{1-\delta}}{1-\delta} $ holds for every $k$ from 1 to $(n-1)$ and then consider the general case. We will use the Bernoulli's inequality, i.e., for real numbers $x,r$, where $0 \le  r \le 1$ and $x \ge -1$, we have $(1+x)^{r}\leq 1+rx$.}

\[
\begin{aligned}
    T(n) &= T(n-n^\delta)+1
         = \left(\dfrac{d}{1-\delta}\right)(n-n^\delta)^{1-\delta}+1\\
         &\le \left(\dfrac{d}{1-\delta}\right)n^{1-\delta}(1 - (1-\delta) n^{\delta-1})+1,  &&\text{ using Bernoulli's inequality}\\ 
         &= \left(\dfrac{dn^{1-\delta}}{1-\delta}\right) - d+1\\
         &\le \left(\dfrac{dn^{1-\delta}}{1-\delta}\right),  &&\text{ for any $d\ge 1$.}         
\end{aligned}
\]

 


\J{
This concludes the prove that $T(n)\le \frac{dn^{1-\delta}}{1-\delta}$, where  \J{$d =  2(1-\delta) + 1$}.
}
 %
\end{proof}




\begin{lemma}\label{lem:ext}
Given a \J{sequence}  $S$ of $n$ integer $k$-tuples, where $k\in o(\log \log n)$,  one can partition $S$ into $O(\frac{n^{1-\delta^k}}{1-\delta^k})$  disjoint subsequences \J{of $k$-tuples} such that each subsequence \J{is monotonic}. 
\end{lemma}
\begin{proof}
 We first extract $n^\delta$ $k$-tuples such that the integer sequence  in the first dimension \J{is monotonic}. 
 We then repeat the process for every dimension, each time on the newly selected $k$-tuples. Therefore, one can  find a subsequence of at least $n^{\delta^\kappa}$ $k$-tuples such that
 they \J{are monotonic}. 
 
Since a single $k$-tuple is trivially \J{monotonic}, 
and since the term $n^{\delta^\kappa}$ is not a constant for $k\in o(\log \log n)$, we can apply Lemma~\ref{lem:recurse} with subsequence size $n^{\delta^\kappa}$. Consequently, we can partition $S$ into $O(\frac{n^{1-\delta^k}}{1-\delta^k})$ disjoint subsequences \J{of $k$-tuples}, \J{which are monotonic}. 
\end{proof}
  
\J{We will also use the following properties of ordered sets of real numbers and integers.}

 \begin{theorem}[Erd\H{o}s--Szekeres~\cite{classic}]\label{thm:erdos}
 Given an ordered set of $(n^2+1)$ distinct real numbers, one can always choose a monotonic subsequence of size at least \J{$n+1$}.
 \end{theorem}

 \begin{lemma} \label{lem:simple}
 Given an ordered multiset $S$ of $n$ integers, it is always possible to choose a monotonic subsequence of size at least $\sqrt{n}$.
 \end{lemma}
\begin{proof}
Let $I$ be the set of distinct integers in $S$. For each integer $i\in I$, let $S_i$ be the ordered subset consisting of all occurrences of $i$ in $S$. Let $x$ be the $j$th element of $S_i$. Add the number $j/(|S_i|+1)$ to $x$. After processing all the elements of $I$, we obtain a set $S'$ of $n$ distinct numbers. By Erd\H{o}s--Szekeres  theorem~\cite{classic}, we now can choose a monotonic subsequence of size at least $\sqrt{n}$ from $S'$. 
\end{proof} 

\section{Simultaneous Embedding of Colored Graphs}
\label{sec:draw}
In this section we describe our result on colored simultaneous embedding. By Remark~\ref{rem:1}, it suffices to concentrate on computing uphill drawings of the spinal paths with low bend complexity. 


\begin{lemma}\label{lem:fuzzy}
Let $Q$ be a set of $k$ spinal paths, each with $n$ vertices. Assume that the paths are compatibly colored with $n^\alpha$ colors, where $0\le \alpha \le  1$. Then one can simultaneously draw $Q$ such that each path is drawn uphill with bend complexity  $O(\min\{n^\alpha, n^{1-(1/2)^k}\})$.
\end{lemma}
\begin{proof}
We consider the following two cases.   

\textbf{Case 1 ($0\le \alpha \le 1/2)$:} In this case we construct the vertex locations on a horizontal line and assign the vertices of the same color a subset of contiguous locations. We now show that every path admits an uphill drawing with bend complexity $O(n^\alpha)$   on these vertex locations. Let $P=\{x_1,x_2,\ldots,x_n\}$ be a path in $Q$ and let $C(x_i)$, where $1\le i\le n$, be the contiguous vertex locations for the color of $x_i$. We compute an uphill drawing of $P$, as follows. We map the vertex $x_1$ to the leftmost vertex location of $C(x_1)$. For each $i$ from $2$ to $n$, we map $x_i$ to the   leftmost available vertex location of $C(x_i)$. We then draw the edge  $(x_{i-1},x_i)$ with an $x$-monotone polyline $L$ such that $L$ does not create any edge crossing with the drawing of $x_1\ldots x_{i-1}$, and the unmapped vertex locations lie above $L$.

Fig.~\ref{fig:spinal} illustrates the construction, where the mapped locations for each color are shown in shaded rectangles. Since we always map a new vertex $v$ to the leftmost available location in set $C(v)$, the mapped locations for each color occupy a contiguous subset of vertex locations. We can use this property, to route an edge from one color set $C(v)$ to another color set $C(w)$, using $O(b)$ bends, where $b$ is the number of distinct colors between these two sets. We need $O(1)$ bends to skip each intermediate color set (e.g., see the bold edge in Fig.~\ref{fig:spinal}(f), where the intermediate color sets are  skipped by `jumping' over  the shaded rectangles). Since there are at most $n^\alpha$ colors, one can construct $L$ with $O(n^\alpha)$ bends. 
\begin{figure}[pt]
\includegraphics[width=\textwidth]{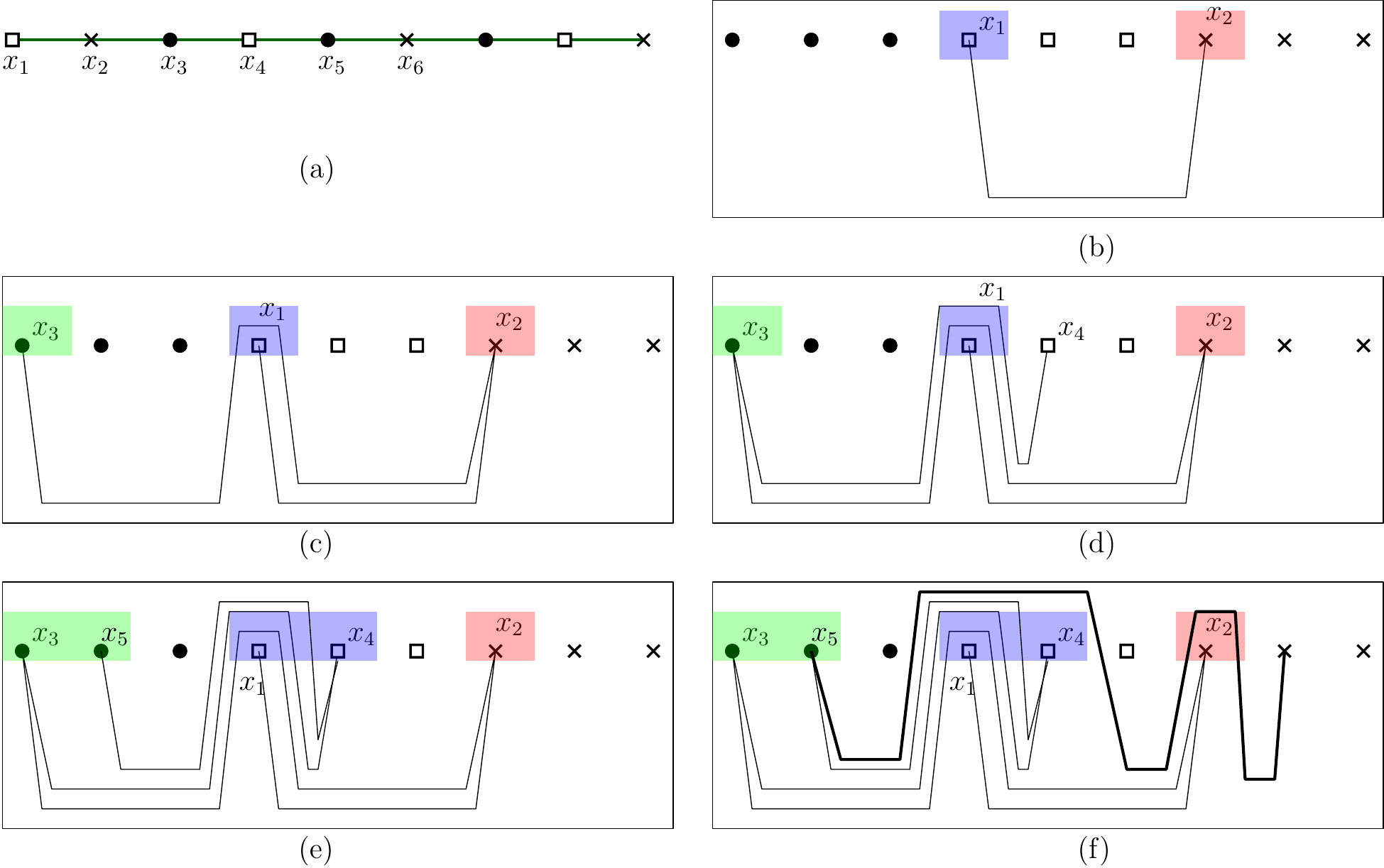}
\caption{(a) A spinal path $P$. (b)--(f)  Construction of an uphill drawing of $P$. }
\label{fig:spinal}
\end{figure}

\textbf{Case 2 ($1/2<\alpha\le 1$):} We first label the vertices of each spinal path  with unique integers from $1$ to $n$  such that for every label $\ell$, the vertices with label $\ell$ in all the paths have the same color (Fig.~\ref{fig:spinal2}(a)). We then create a set $I$ of $n$ integer $k$-tuples, each of size $k$. The $i$th entry of the $j$th $k$-tuple, where $1\le i\le k$ and $1\le j\le n$, is the position of the $j$th label in the $i$th spinal path.  Fig.~\ref{fig:spinal2}(b) illustrates the construction of  $I$ for the paths of  Fig.~\ref{fig:spinal2}(a). Note that the colors of the vertices do not play any role in this construction. 
\begin{figure}[pt]
\includegraphics[width=\textwidth]{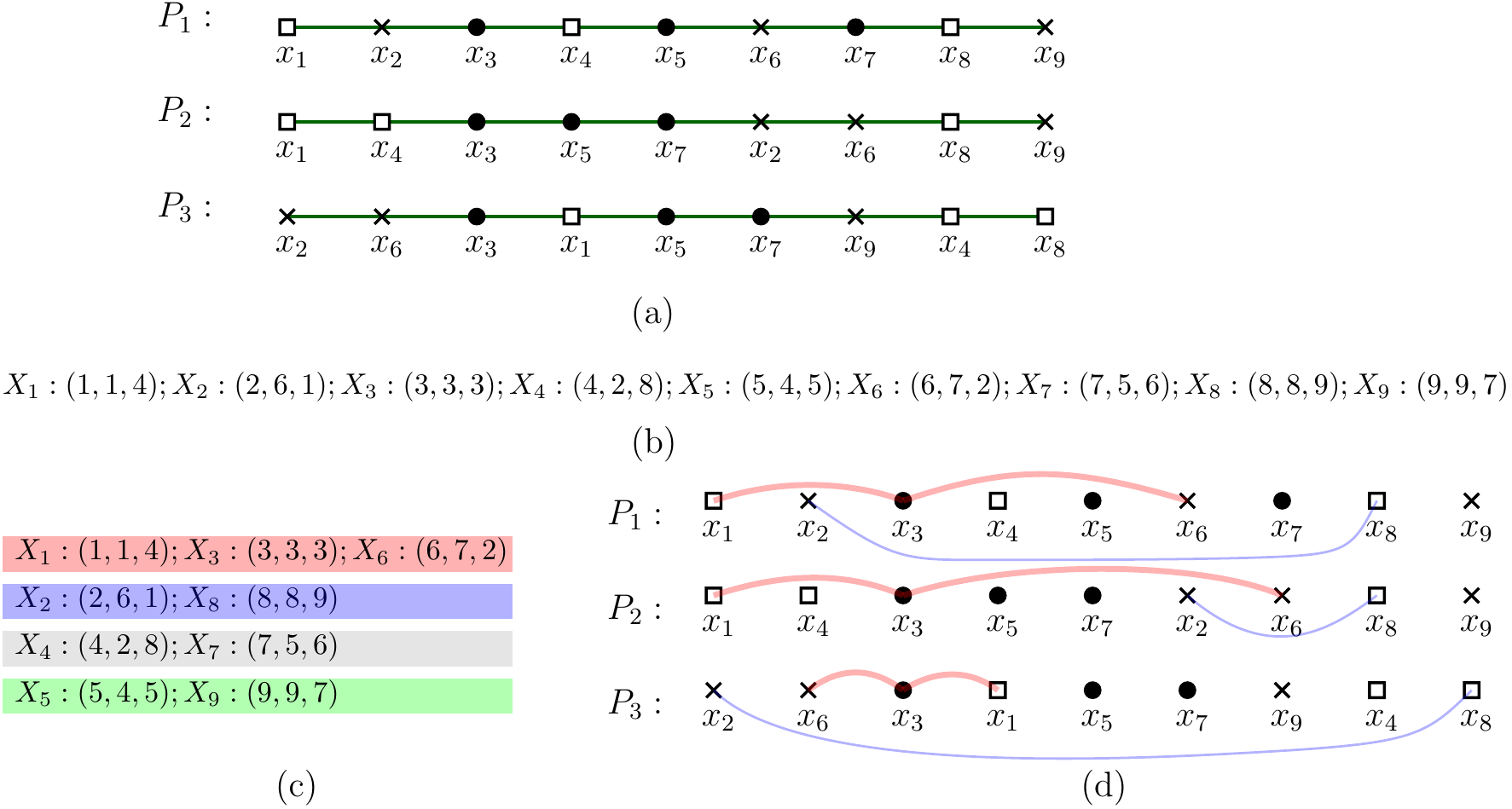}
\caption{(a) A set of spinal paths.  (b) Construction of $I$. (c) Partition of the paths into monotonic sequences. (d) Illustration for the monotonicity for the first two elements of the partition.}
\label{fig:spinal2}
\end{figure}

By Lemma~\ref{lem:simple}, for any multiset of $n$ integers, we can find a monotonic subsequence of length at least  $n^{1/2}$. By Lemma~\ref{lem:ext}, one can partition $I$ into  $O\left(\frac{n^{1-(1/2)^k}}{1-(1/2)^k}\right) \in O(n^{1-(1/2)^k})$ disjoint subsequences, where each subsequence is monotonic. Fig.~\ref{fig:spinal2}(c)--(d)  illustrate such a partition for the sequence $I$ of  Fig.~\ref{fig:spinal2}(b). 

 We now construct the vertex locations on a horizontal line and assign the vertices of the same subsequence a subset of contiguous locations. Fig.~\ref{fig:complex}(b)  illustrates the point set determined by the monotonic sequences of Fig.~\ref{fig:spinal2}(c). \J{The monotonic sequences are $(X_1,X_3,X_6), (X_2,X_8), (X_4,X_7)$ and $(X_5,X_9)$.} We now show that every path admits an uphill drawing with $O(n^{1-(1/2)^k})$ bends per edge on these vertex locations.

Let $P=\{x_1,x_2,\ldots,x_n\}$ be a path in $Q$ and let $C'(x_i)$, where $1\le i\le n$, be the contiguous vertex locations for the subsequence containing $x_i$. We compute an uphill drawing of $P$, as follows. \J{We first map the vertex $x_1$ to its corresponding vertex location $X_1$. For each $i$ from $2$ to $n$, we first map $x_i$ to its corresponding vertex location $X_i$, and then   draw the edge  $(x_{i-1},x_i)$ with an $x$-monotone polyline $L$ such that $L$ does not create any edge crossing with the drawing of $x_1\ldots x_i$, and the unmapped  locations lie above $L$.} 
\begin{figure}[h]
\includegraphics[width=\textwidth]{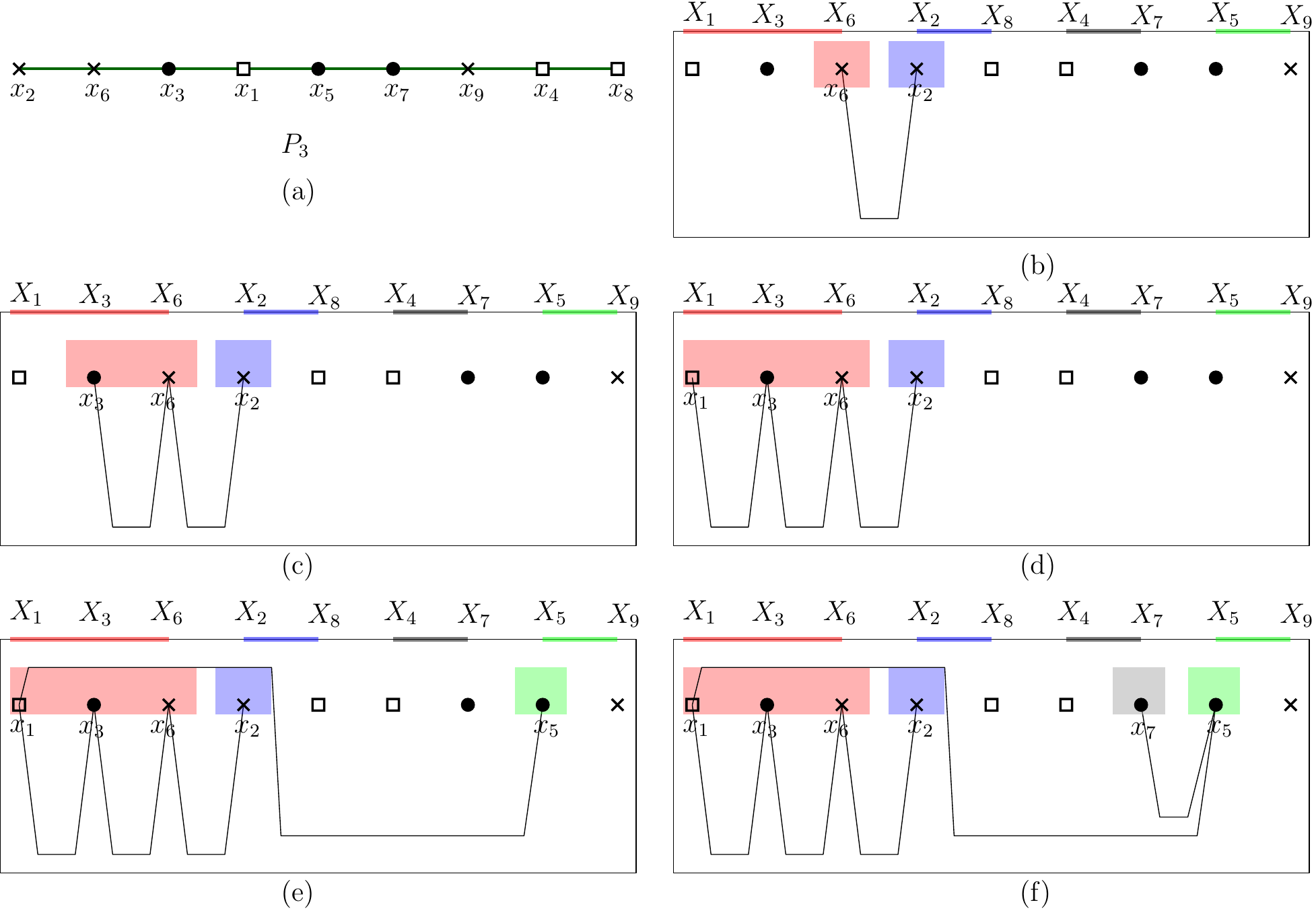}
\caption{(a) The spinal path $P_3$ of Fig.~\ref{fig:spinal2}, and (b)--(f) it's uphill drawing on the point set determined by the monotonic sequences.  }
\label{fig:complex}
\end{figure}
%

The monotonicity of the set $C'(x_i)$ ensures that the  mapped locations of $C'(x_i)$ remain  consecutive. Therefore, 
one can route the edge from one monotonic sequence $C'(x_{i-1})$ to another monotonic sequence $C'(x_i)$, using $O(b)$ bends, where $b$ is the number of distinct monotonic    sequences between these two monotonic sequences.
 Fig.~\ref{fig:complex}(b)--(f) illustrate the  computation of an uphill drawing for the path of Fig.~\ref{fig:complex}(a). The mapped locations for each monotonic sequence are enclosed in shaded rectangles.

We need $O(1)$ bends to skip each intermediate monotonic sequence (e.g., the intermediate monotonic sequences can be skipped by `jumping' over  the shaded rectangles). Since there are at most $O(n^{1-(1/2)^k})$  monotonic sequences, we can construct $L$ with $O(n^{1-(1/2)^k})$ bends. 
\end{proof}

\begin{figure}[h]
\includegraphics[width=.9\textwidth]{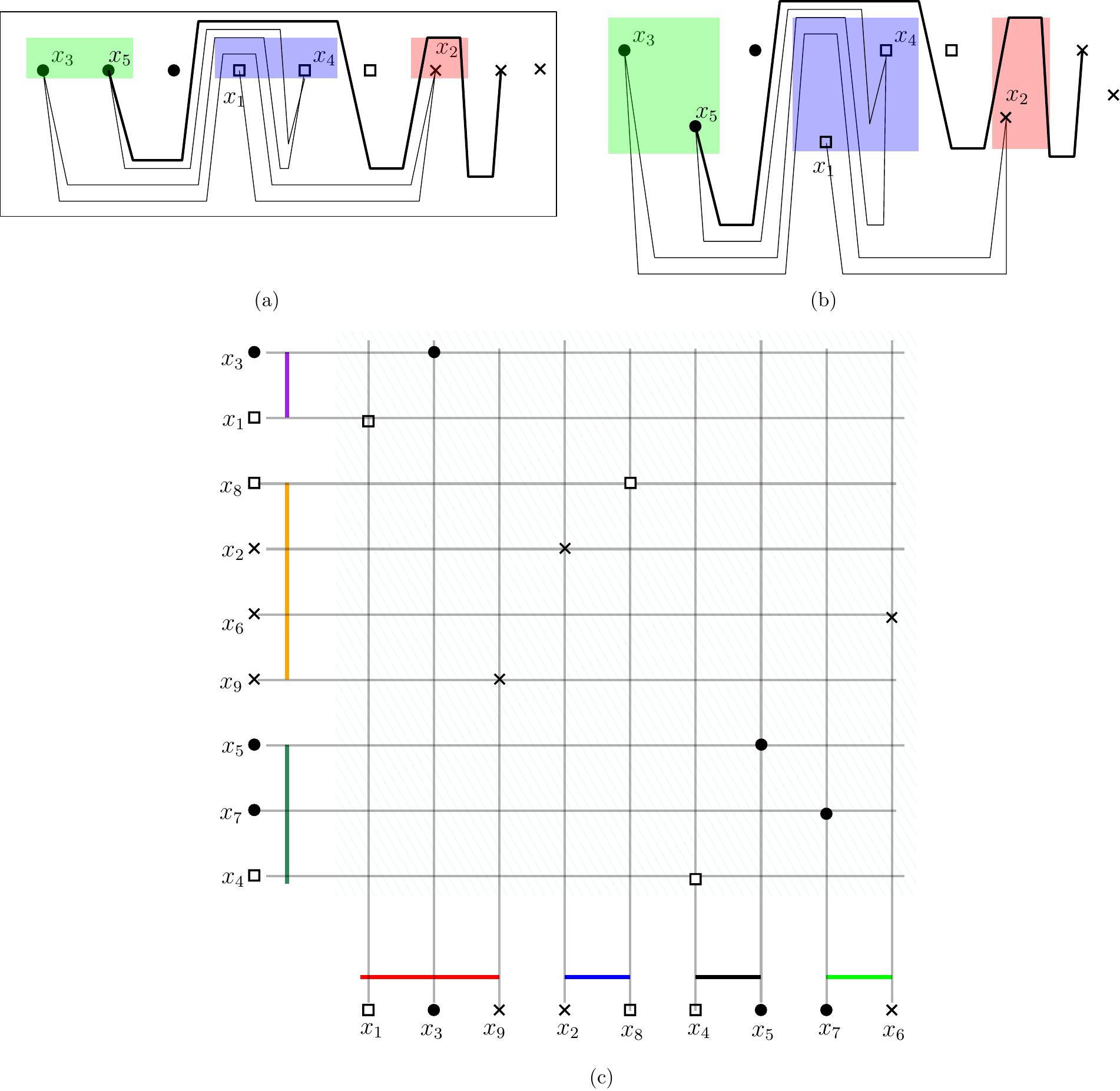}
\caption{(a)--(b) Illustration for computing an uphill drawing when the points are not on a line. (c) Construction of the point set $S$. }
\label{fig:half}
\end{figure}

Lemma~\ref{lem:fuzzy} and Remark~\ref{rem:1} together implies that any set of $k$ graphs, compatibly colored with $c$ colors, admits a simultaneous embedding with bend complexity  $O(\min\{c, n^{1-(1/2)^k}\})$. This can be  \J{further improved} by the following 
 observations (Obs 1-2). \J{These observations are made in~\cite{DBLP:journals/siamdm/DurocherM18}, which generalize a standard technique of computing simultaneous geometric embedding of two paths~\cite{DBLP:journals/comgeo/BrassCDEEIKLM07} to drawing any number of paths with bends.} 
\begin{description}
    \item[Obs1.] The vertex locations of an uphill drawing can be moved vertically and the edges can be redrawn such that the new drawing is also uphill with the same bend complexity. Fig.~\ref{fig:half} (a)--(b) illustrate such an example.
    \item[Obs2.] Half of the spinal paths can be drawn along the x-axis, and the other half along the y-axis. 
\end{description}
\J{While Obs 1 is straightforward to verify, we  give a brief overview of Obs 2 to make the paper self-contained.}

Given a set $Q$ of $k$ spinal paths, first partition them into two sets $Q_1$ and $Q_2$ such that  $Q_1$ contains $\lceil k/2\rceil$ spinal  paths, and $Q_2$ the rest of the spinal paths. Let $I_j$, where $1\le j\le 2$, be the set of integer $k$-tuples constructed for $Q_j$, and let $M_j$ be the partition of $I_j$ into monotonic sequences. Now construct a point set along the x-axis using $M_1$, and another point set along the y-axis using $M_2$ (Fig.~\ref{fig:half}(c)). The final point set is constructed by taking for each label, the point with $x$ and $y$-coordinates equal to the positions  of the label along the x-axis and y-axis, respectively.  Fig.~\ref{fig:half}(c) depicts the final point set $S$ on an integer  grid. \J{We have two monotonic sequences,  one along x-axis and the other along y-axis. We now draw $Q_1$  along the x-axis, and $Q_2$  along the y-axis. The  bend complexity of the drawing of  $Q_1$ does not interfere with bend complexity of the drawing of $Q_2$. Therefore, the bend complexity of the resulting   drawing are determined by only half of the spinal paths. 
}  

By Lemma~\ref{lem:fuzzy} and Remark~\ref{rem:1}, we now have the following theorem.
 %

\begin{theorem}
\label{thm:specific}
Given  $k$ compatibly colored planar graphs, each with $n$ vertices, one can compute a simultaneous embedding of these graphs with bend complexity  $O(\min\{c,n^{1-1/\gamma}\})$, where $c$ is the   number of distinct colors  and $\gamma = 2^{\lceil k/2 \rceil}$.  
\end{theorem} 

\section{Trading Bend Complexity with Vertex Locations}\label{sec:simple}
In this section we show that allowing more than $n$ vertex locations help improve the bend complexity. In particular, $n\lceil c/b \rceil$ vertex locations, where $b\ge 1$,   suffice to embed any set of compatibly colored $n$-vertex planar graphs with bend complexity $O(b)$, where $c$ is the number of colors.

Let $Q$ be a set of compatibly colored spinal paths, each with $n$ vertices. We first partition the colors into at most $b$ disjoint sets  $C_1,C_2,\ldots C_b$ such that each contains at most $\lceil c/b \rceil$ colors. For each $1\le j\le b$, let $N_j$ be the number of vertices,  \J{whose color belong  to} $C_j$ (in the spinal path). Then $n = \sum_{1\le j\le b} N_j$.
 By a \emph{chain of $C_j$}, we denote a point set of $|C_j|$ points, each colored by a distinct color from $C_j$.  We now create a point set $S$ of size $n\lceil c/b \rceil$ along the x-axis, as follows.


The points in $S$ are arranged into  $b$ subsets $S_1,\ldots,S_b$, each contains a set of contiguous points of $S$. The set $S_j$  consists of $N_j$ chains of $C_j$, and thus $N_j|C_j| = N_j\lceil c/b \rceil$ points. Fig.~\ref{fig:color} illustrates such an example with 5 colors. We now show that the uphill drawing of a spinal path can be computed with $O(b)$ bends per edge.

\begin{figure}[h]
\includegraphics[width=\textwidth]{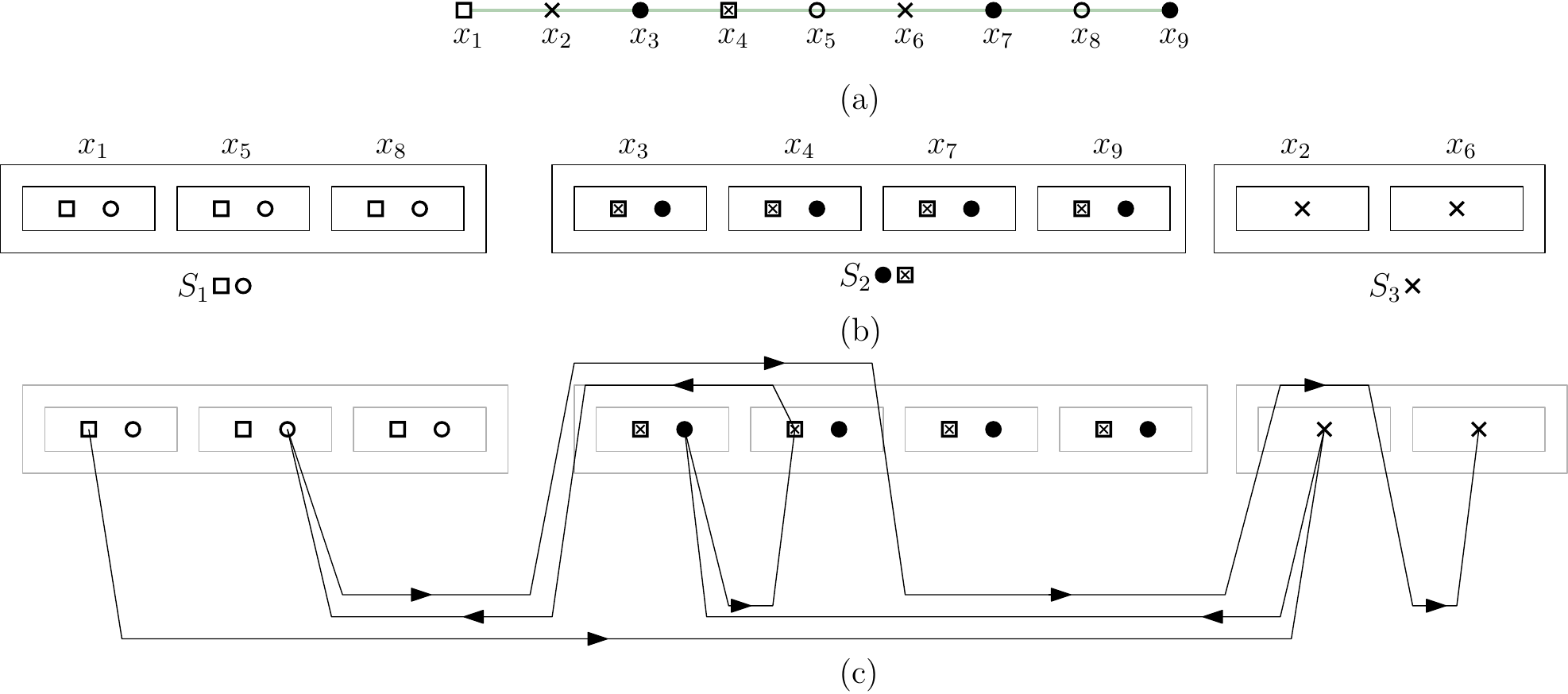}
\caption{(a) A spinal path $P$ to illustrate $N_j$ for different colors. (b) Construction of $S$, with a mapping for the vertices of $P$.   (c) Illustration for an uphill drawing of $P$ on $S$. }
\label{fig:color}
\end{figure}

For each $C_j$, we map the corresponding vertices of the spinal path from left to right on the chains of $S_j$. We now construct the drawing following the same approach as in the proof of Lemma~\ref{lem:fuzzy}, where we treat each chain as a single vertex location (i.e., from each chain, we use only one point for mapping). At each step of the construction, the mapped chains of each subset $S_j$ remain consecutive. Since there are at most $b$ such subsets, the edges can be routed with $O(b)$ bends per edge.   By Remark~\ref{rem:1}, we now have the following theorem.
  
\begin{theorem}
\label{thm:specific}
Given a set of compatibly colored planar graphs, each with $n$ vertices, there exists an $O(b)$-bend universal point set  of $n\lceil c/b\rceil$ points for these graphs, where $c$ is the number of distinct colors in the input and $b\ge 1$. 
\end{theorem} 

\section{Conclusions and Directions for Future Work}
\label{sec:conclusion} 
In this paper we show that colored simultaneous embedding of $k\in o(\log \log n)$ compatibly colored graphs, each with $n$ vertices, can be computed on   $n$ vertex locations and with sublinear bend complexity. The running time of our approach is polynomial in $n$ and $k$, which is determined by the time for partitioning   an ordered set of $n$ integer $k$-tuples into monotonic subsequences (Lemma~\ref{lem:ext}). We also show that any number of $n$-vertex compatibly colored graphs can be embedded on a set of $n\lceil c/b\rceil$ points with bend complexity $O(b)$, where $c$ is the number of colors and $b\ge 1$.         

The existence of a monotonic subsequence of size $\sqrt{n}$ in any integer sequence was the key to improve the bend complexity in  Section~\ref{sec:draw}. However, in a colored simultaneous embedding, the corresponding sequences are integer multisets. Hence we attempted to generalize  
Erd\H{o}s--Szekeres theorem~\cite{classic} to compute large monotonic sequences in a   multiset of positive integers. We were able to prove the 
existence of a monotonic subsequence of size  $\max\{\sqrt{n},  \sqrt{c} +\frac{n}{c}-2\}$ (Lemma~\ref{lem:main}), which is larger than $\sqrt{n}$ only when the number of distinct integers $c$ is smaller than $\sqrt{n}$. It would be interesting to investigate multisets where the number of distinct integers is larger than $\sqrt{n}$.

\begin{lemma}
\label{lem:main}
Let $S$ be an ordered multiset of $n$ integers, where  $n = kc+1$, $c$ is the number of distinct integers in $S$, and $k$ is a positive integer. 
Assume that $c$ is a perfect square. Then $S$ contains a monotonic   subsequence of size at least $\max\{\sqrt{n}, \sqrt{c} +\frac{n}{c}-2\}$. 
\end{lemma}
\begin{proof}
If $S$ does not contain $c$ distinct integers, then we can apply the same technique as in the proof of Lemma~\ref{lem:simple} to transform $S$ into a multiset that satisfies this property. Hence it suffices to prove the claim when $S$ contains $c$ distinct integers. Since the $\sqrt{n}$ lower bound is guaranteed by Lemma~\ref{lem:simple}, we focus on the other term. 

We apply an induction on $k$. If $k=1$, then we have $c+1$ integers, and by  Erd\H{o}s-Szekeres theorem~\cite{classic}, we can find a monotonic subsequence of size at least $\sqrt{c}$. Since $n=c+1$, we have $\sqrt{c}  \ge  \sqrt{c} + \frac{n}{c}-2$. We may thus assume that $k>1$, and for each integer smaller than $k$, the claim holds. 

By Lemma~\ref{lem:simple}, we can find a monotonic subsequence \J{L} of size $\sqrt{c}$ from the first $c+1$ integers and delete the last element \J{$\ell$} of the subsequence.  We choose the earliest possible subsequence, i.e., a monotonic subsequence that minimizes the  position of the last element. \J{Later, we will refer to this last element $\ell$ as the \emph{representative} of the monotonic sequence $L$.} Note that we can repeat this step for at least $(n-c)$ steps, and remove a set $X=\{x_1,\ldots,x_{n-c}\}$ of $(n-c)$ integers. The set $X$ is ordered, i.e., $x_i$, where $1\le i\le (n-c)$, denotes the element deleted at step $i$. 
 
 Let \J{$q$ be an} integer with the maximum frequency in $X$, and let $X_q\subseteq X$ be the set consisting of all $q$. Since there are  at most $c$ distinct integers in $X$, $|X_q| \ge \lceil (n-c)/c \rceil$. Consider the earliest integer $p$ in  $X_q$. The integer  $p$ is a representative of a monotonic subsequence $Q$ of size $\sqrt{c}$. The subsequence consisting of $Q$  followed by  $X_q\setminus\{p\}$ is the required subsequence of size at least 
$
  \sqrt{c}+\left\lceil\frac{  (n-c)}{c}\right\rceil-1     
    \ge   \sqrt{c} + \frac{n}{c}-2.$
\end{proof}

 

Our work raises several natural directions for future research, e.g., (a) Prove a tight bound on the length of the monotonic subsequence that is guaranteed  to exist in every multiset of $c$ distinct numbers. (b) Prove a non-trivial lower bound  on the bend complexity for colored  simultaneous embedding. (c) Improve our results on the interplay between the bend complexity and the number of vertex locations. 

\J{\subsubsection*{Acknowledgement.} The author  thanks the anonymous reviewers for their helpful  comments, which improved the presentation of the paper.} 

\bibliographystyle{abbrv}
\bibliography{ref}

\end{document}